\documentclass{amsart}[11pt]
\usepackage{amsmath, amsfonts, amsthm, amssymb,mathtools}
\usepackage{subfigure}
\usepackage{stmaryrd}
\usepackage{verbatim}
\usepackage{hyperref}
\usepackage{color}
\usepackage{ulem}
\usepackage[dvips]{graphicx}
\usepackage{enumerate}
\numberwithin{equation}{section}
\mathtoolsset{showonlyrefs=true}
\newtheorem{theorem}{Theorem}[section]
\newtheorem{corollary}[theorem]{Corollary}
\newtheorem{lemma}[theorem]{Lemma}

\theoremstyle{definition}
\newtheorem{definition}[theorem]{Definition}
\newtheorem{remark}[theorem]{Remark}

\newcommand{\ind}{1\hspace{-2.1mm}{1}} 
\newcommand{\RR}{\mathbb{R}}

\newcommand{\QQ}{\mathbb{Q}}
\newcommand{\EE}{\mathbb{E}}

\newcommand{\Oo}{\mathcal{O}}
\newcommand{\Pp}{\mathcal{P}}
\newcommand{\Ww}{\mathcal{W}}

\newcommand{\D}{\mathrm{d}}
\newcommand{\E}{\mathrm{e}}
\newcommand{\I}{\mathrm{i}}
\newcommand{\BS}{\mathrm{BS}}
\newcommand{\eps}{\varepsilon}
\newcommand{\SBar}{\overline{S}}
\newcommand{\ts}{S}

\newcommand{\TT}{[0,T]}
\newcommand{\pf}{\mathfrak{p}}
\newcommand{\qfe}{\qf_{\eps}}
\newcommand{\qf}{\mathfrak{q}}
\newcommand{\df}{\mathfrak{d}}
\newcommand{\ff}{\mathfrak{f}}
\newcommand{\gf}{\ff^{\leftarrow}}
\newcommand{\hf}{\mathfrak{h}}
\newcommand{\Vf}{\mathfrak{V}}
\newcommand{\dfs}{\df(x,\sigma)}
\newcommand{\Cbs}{\mathtt{C}_{\BS}}
\newcommand{\Pbs}{\mathtt{P}_{\BS}}
\newcommand{\Cf}{\mathtt{C}}
\newcommand{\Pf}{\mathtt{P}}
\newcommand{\Pft}{\Pf}
\newcommand{\Iv}{\mathrm{I}}

\title{The Log Moment formula for implied volatility}
\author{Vimal Raval}
\email{v.raval@gmail.com}
\author{Antoine Jacquier}
\address{Department of Mathematics, Imperial College London and the Alan Turing Institute}
\email{a.jacquier@imperial.ac.uk}
\date{\today}
\thanks{The authors would like to thank Masaaki Fukasawa for insightful comments, in particular leading to Section~\ref{sec:FukExtension}}

\begin{document}
\maketitle
\begin{center}
\textit{Dedicated to the memory of Mark H.A Davis}
\end{center}

\begin{abstract}
We revisit the foundational Moment Formula proved by Roger Lee fifteen years ago.
We show that when the underlying stock price martingale admits finite log-moments $\EE\left[|\log S_t|^q\right]$ for some positive $q$,
the arbitrage-free growth in the left wing of the implied volatility smile is less constrained than Lee's bound. The result 
is rationalised by a market trading discretely monitored variance swaps wherein the payoff is a function of squared log-returns, 
and requires no assumption for the underlying martingale to admit any negative moment. In this respect, the result can derived from a model-independent setup.
As a byproduct, we relax the moment assumptions on the stock price
to provide a new proof of the notorious Gatheral-Fukasawa formula expressing variance swaps in terms of the implied volatility.
\end{abstract}

\section{Introduction}
Implied volatility is at the very core of Quantitative Finance
and is the day-to-day gadget that traders observe and manipulate.
The increasing complexity of stochastic models we have witnessed over the past thirty years
is a testimony to its importance and subtlety.
One key issue is the absence of closed-form expression for the latter, 
leaving it to the sometimes capricious moods of numerical analysis.
Among the plethora of research in this direction, carried out both by academics and by practitioners,
model-free results, with minimum assumptions, are scarce.
Roger Lee's Moment Formula~\cite{Lee} was a groundbreaking result 
and its importance cannot be understated:
it provides a direct link between the slope of the smile in the wings and the moments of the distribution of the underlying asset price.
It serves not only to infer directly observed information about the implied volatility smile into constraints on 
model parameters but also to provide arbitrage-free solutions to the extrapolation problem
(how to evaluate options for strikes outside the observed range).
Recent refinements have led to a deeper understanding 
of the information contained in the implied volatility smile, 
determining whether the probability of default of the underlying could be inferred~\cite{DeMarco} or 
the potential lack of martingality of the latter~\cite{JKR}.
These have complemented the otherwise exhaustive literature on the asymptotic behaviour of 
stochastic models in Finance, a thorough review of which can be consulted in~\cite{AsymptBook}.

Asymptotic methods have both supporters and enemies, 
the former trying to expand the abundance of techniques to every possible model,
while the latter sometimes dismiss the usefulness of these results.
The truth as often lies somewhere in the middle but asymptotic results nevertheless provide useful information 
about the qualities and pitfalls of models with regards to real-life practices. 
With this in mind, we revisit Lee's formula when presented with 
some underlying stock price, the prices of finitely many co-maturing European Call and Put options as well as 
a variance swap with the same maturity.  
In ~\cite{dh05} conditions were stated under which a given set of European Call and Put prices all maturing at the same time~$T$ is consistent with absence of arbitrage,
which is shown to be equivalent to the existence of a \textit{market model}; a filtered proability space carrying an adapted, integrable process $(S_t)_{t\in[0,T]}$, with $S_0$ equal to the time-0 stock price,
in which the discounted stock price process is a martingale and the discounted expectations of the Put option payoffs recover the observed time-zero values. 
In~\cite{DOR14} robust model-free conditions are provided, when the process $(S_t)_{t\in[0,T]}$ admits continuous sample paths, for a set of European Put option prices and 
 \textit{continuously monitored} variance swap price to be consistent with absence of arbitrage. 
Our approach here is to make no assumption on the dynamics of the stock-price price process, and to instead infer
limiting behaviour of the left-wing given merely the information that the marginal distribution of $S_T$ admits finite log moments, $\EE\left[|\log S_T|^q\right]$ for finite positive moments~$q$,
which is motivated directly from the market since it trades a (discretely or continuously monitored) variance swap.
Further, it is feasible for market models to exist that do not admit any negative moment for the stock price and 
Lee's Moment formula (on~$S_T$ rather than $\log(S_T)$) implies that the left wing of the smile has slope precisely 
equal to two.
Our newly formulated Log-Moment formula allows us to provide higher-order term in this asymptotic behaviour,
fully characterised by the moments of the log-stock price.

We provide a precise formulation of the problem and a thorough review of Roger Lee's 
Moment formula in Section~\ref{sec:problemformulation} 
before stating and proving the new Log Moment formula in Section~\ref{sec:results}.
As a byproduct, we revisit the Fukasawa-Gatheral formula expressing variance swaps
in terms of the implied volatility and provide a new proof with relaxed assumptions
and further show how this improves Fukasawa's representation~\cite{Fukasawa} of option prices in terms of implied volatility.
We highlight in Section~\ref{sec:applications} a few stochastic models, 
both with continuous paths and with jumps, used in Finance for which our 
formula refines Lee's standards.

\section{Problem Formulation and background}\label{sec:problemformulation}
We consider a time horizon $\TT$, with $T>0$ and a filtered probability space $\left(\Omega, \mathcal{F}, (\mathcal{F})_{t \in\TT}, \QQ \right)$ satisfying the usual assumptions and carrying two adapted processes $(\SBar_t)_{t \in \TT}$, the asset price, 
and $(B_t)_{t \in \TT}$, starting from $B_0=1$, where~$B_T$ represents the value at time~$T$ of~\pounds$1$ 
invested at time~$0$ in the money-market account.
We further denote by~$F_t$ the $t$-forward price of~$\SBar$ from time~$0$, 
thus $F_0=\SBar_0$ is the observed spot price.  
Dividends may be paid by the asset~$\SBar$, but we do not make any assumptions about these. 
The process $(\SBar_t)_{t \in \TT}$ is assumed to be a strictly positive $\QQ$-semimartingale.
We finally assume the existence of a zero-coupon-bond maturing at~$T$ with face-value \pounds $1$, traded with price
$P_T = \EE[B_T^{-1}]$ so that 
$F_T = \EE[\SBar_T]$,
where all expectations are taken under~$\QQ$.
We consider the setup where
$(\SBar_t)_{t\in\TT}$ can be traded with no transaction costs
and where the interest rate for borrowing and lending is the same, however not necessarily deterministic.
Here, the probability~$\QQ$ thus plays the role of the $T$-forward measure.
By no-arbitrage arguments, a vanilla Put option with strike~$K$ is worth
$\Pf_0(K) := P_T\EE[(K -\SBar_T)^+]$, for $K\geq 0$.
Using normalised units for the stock-price $\ts_T := \SBar_T/F_T$ and log-moneyness $x := \log(K/F_T)$, 
the normalised price of the Put option is denoted by
$$
\Pft(x) : = \frac{\Pf_0(K)}{P_T F_T} = \EE\left[\left(\E^x - \ts_T\right)^+\right].
$$
Recall now the Black-Scholes formula for the European Put option:
\begin{equation}\label{eq:Put} 
\Pbs(x, \sigma) = \E^x \Phi[-\dfs] - \Phi[-\dfs - \sigma],
\end{equation}
where~$\Phi$ denotes the Gaussian cumulative distribution function and 
\begin{equation}\label{eq:d}
\dfs := -\frac{x}{\sigma} - \frac{\sigma}{2},
\end{equation}
which is nothing else than the usual $d_2$ or $d_-$ from the Black-Scholes formula.
Since the maturity~$T$ is fixed throughout the whole paper, 
we work with normalised volatility~$\sigma$ rather than the classical $\sigma\sqrt{T}$ notation.
This has the clear advantage of avoiding cluttered statements.
\begin{definition}\label{def:iv}
For any log-moneyness $x \in \RR$, the implied volatility $\Iv(x) \in [0,\infty)$ is the unique non-negative solution to
$\Pft(x) = \Pbs(x, \Iv(x))$.
\end{definition}

The implied volatility~$\Iv(x)$ is well defined whenever 
$\Pft(x) \in [(\E^x-1)^+, \E^x]$, which holds since the stock price is a true martingale.
Our starting point is the following initial bound for the implied volatility~\cite[Lemma 3.3]{Lee}:

\begin{lemma}\label{applem:leebound}
For any $\beta>2$, there exists $x^*<0$ such that 
$\Iv(x) < \sqrt{\beta |x|}$ for all $x < x^*$.
 For $\beta=2$, the same holds if and only if $\QQ[\ts_T=0]<\frac{1}{2}$.
\end{lemma}

In our setup, $\ts_T$ is strictly positive almost surely and therefore
 $\Iv(x)= \Oo(\sqrt{2|x|})$ as~$x$ tends to $-\infty$.
When the number of finite inverse power moments for the stock price is known, the small-strike
Moment Formula due to Lee~\cite{Lee} refines the above result:

\begin{theorem}[Lee's Left Moment Formula] \label{thm:Lee}
Let
$$
\pf:= \sup\left\{ p>0: \EE\left[\ts_T^{-p}\right]< \infty \right\}
  \quad \text{and} \quad \beta_L:= \limsup_{x \downarrow-\infty} \frac{\Iv(x)^2}{|x|}.
$$
Then $\beta_L \in [0,2]$, $ \pf= \frac{1}{2 \beta_L} + \frac{\beta_L}{8} - \frac{1}{2}$,
with $\frac{1}{0}:=\infty$. 
Equivalently,
$ \beta_L = 2 - 4(\sqrt{\pf^2 + \pf} - \pf)$, 
equal to~$0$ for $\pf=\infty$.
\end{theorem}

This theorem was one of the first model-free result about the relationship between 
the distribution of the stock price and the behaviour of the implied volatility. 
The $\limsup$ in Lee's result was further strengthened to a genuine limit by Benaim and Friz~\cite{BenaimFriz1, BenaimFriz2}, 
albeit with additional assumptions.
It is really a cornerstone in the implied volatility modelling literature
and has provided academics and practitioners robust consistency checks for extrapolation of the smile.
Lee also proved a symmetric right-wing formula, but we omit its presentation as we shall not require it here.
This left-wing behaviour of the smile left two unresolved issues however: 
if $\ts_T$ has a strictly positive mass at the origin, then Lee's expression is not able to distinguish it from a mass-less distribution with fat tails; this was tackled in~\cite{DeMarco}.
The second issue is that in fact no information about the moments of~$\ts_T$ is really available in the market,
and the so-called Power options~\cite{CarrLeeRobust} are rarely traded.
However, variance swaps are traded on the market and it is thus a natural question
to check if Lee's celebrated Moment Formula could be refined to take into account these highly liquid derivatives.

\section{Variance swaps and the Log-Moment formula}\label{sec:results}
\subsection{Characterisation of variance swaps}
Variance swaps are highly liquid traded derivatives on the Equity market. One can describe them as a standard swap, where, over the time horizon~$\TT$ the floating leg is equal to the (annualised) realised variance~$\Vf^d_T$ as measured by
$$
\Vf^d_T := \frac{252}{T}\sum_{i=1}^{n}\log\left(\frac{S_{t_i}}{S_{t_{i-1}}}\right)^2,
$$
in which $0=t_0<t_1<\cdots<t_n=T$ is an equidistant partition with $t_i-t_{i-1} = iT/n$ corresponding to one day for some positive integer~$n$. 
The superscript~$^d$ here refers to the fact that this definition corresponds to the so-called discretely monitored 
variance swap.
The early advances on the hedging and pricing of the variance swap by Neuberger~\cite{neuberger_90}, Dupire~\cite{dupire93},  Carr and Madan~\cite{cm02} and Demeterfi, Derman, Kamal and Zou~\cite{Derman}
led to the instrument being used extensively by traders to express views on future realised variance and hedging volatility risk. These advances hinged on assuming (i) the stock price process is a continuous semi-martingale with strictly positive values, (ii) the 
realised variance is \textit{continuously monitored} and measured by the quadratic variation $\langle \log S\rangle_T$, and (iii) Call or Put options maturing at time $T$ are traded for all strikes $K  \in \RR_+$. It\^o's formula for continuous semi-martingales applied to
$-\log S_T$, then gives
$$
 \langle \log S\rangle_T
 = \int_0^T \frac{\D \langle S\rangle_t}{S_t^2}
 = -2 \log\left(\frac{S_T}{S_0}\right) + 2 \int_0^T \frac{\D S_t}{S_t}.
$$
The variance swap is replicated by holding a contract paying 
$-\log(S_T/S_0)$ and dynamic trading in the underlying stock.  
Now, the log payoff $-\log S_T$ is redundant, since
$$ -\log\left(\frac{S_T}{S_0}\right)
 = \frac{S_T - S_0}{S_0} + \int_0^{S_0}\frac{(K - S_T)^+}{K^2} \D K + \int_{S_0}^{\infty}\frac{(S_T - K)^+}{K^2} \D K,$$
i.e.\ it is hedged by a static position in the underlying asset, 
the continuum of Call and Put options, and cash.  
The variance swap payoff in this setup is therefore fully replicated, with no assumptions on the dynamics of
the price process~$S$, except for continuity.  
It thus follows that the variance swap-rate is the forward cost of the full
hedging portfolio.  
When interest rates are zero and dividends are not paid by the underlying asset, this is
$$ 2\int_0^{S_0}\frac{\Pf_0(K)}{K^2} \D K + 2\int_{S_0}^{\infty}\frac{\Cf_0(K)}{K^2} \D
K,$$
provided both integrals are finite, with~$\Pf_0(K)$ and~$\Cf_0(K)$ the prices of Put and Call options with strike~$K$.  
The subtle impact of jumps on the prices of variance swaps was treated thoroughly by Broadie and Jain~\cite{Broadie}.
In both the discretly monitored and the continuously monitored case, 
the moments of the underlying stock price are not at play, but rather the moments of its logarithm,
thus creating the need to refine Lee's formula to this case.

\subsection{The Log-Moment formula}
Our main result is the following Log-Moment Formula:
\begin{theorem}\label{thm:LMFormula}
Let
$\qf := \sup\left\{q\geq 0: \EE\left[\left|\log \ts_T\right|^q\right] < \infty\right\}$ be finite.
Then
$$
\liminf_{x \downarrow-\infty}\frac{\df(x, \Iv(x))}{\sqrt{2 \log|x|}} = \sqrt{\qf}.
$$
\end{theorem}
It is clear that~$\qf$ does not provide information about the right tail of the distribution of~$\ts_T$.
Since~$\ts$ is a martingale, its first moment is finite and therefore, for any $q\geq 0$,
there exists some constant $c_q\geq 1$ such that
$$
\EE\left[|\log \ts_T|^q\ind_{\{\ts_T\geq c_q\}}\right]
\leq \EE\left[|\ts_T|\ind_{\{\ts_T\geq c_q\}}\right]
\leq \EE\left[\ts_T\right],
$$ 
which is finite, since $\ts_T$ is strictly positive almost surely.
The following corollary is immediate but shows the immediate consequences of the Log-Moment Formula 
on the behaviour of the implied volatility in the left wing.
\begin{corollary}\label{corL:IVExpansion}
In the setting of Theorem~\ref{thm:LMFormula}, at least along a subsequence, we have, as $x\downarrow-\infty$,
\begin{align*}
\Iv(x)  & = \sqrt{2\qf\log(|x|) - 2x} - \sqrt{2\qf\log(|x|)},\\
 & = \sqrt{2|x|} - \sqrt{2\qf\log|x|} + \frac{\qf\log(|x|)}{\sqrt{2|x|}} + \mathcal{O}\left(|x|^{-3/2}\right),
\end{align*}
\end{corollary}

Benaim and Friz~\cite{BenaimFriz1, BenaimFriz2} refined Lee's result, with additional assumptions, 
from a $\liminf / \limsup$ statement to a genuine limit.
One could investigate how this might apply here, but we defer it to a future analysis 
in order not to clutter our main result with extra technical assumptions.
An interesting feature however is the form of the small-strike implied volatility expansion in Corollary~\ref{corL:IVExpansion}.
The slope equal to~$2$ of the total implied variance~$\Iv^2$ is trivial from Lee's result (Theorem~\ref{thm:Lee})
since~$\qf$ finite implies $\pf=0$ (no negative moment of the stock price exists).
Lee's formulation however does not provide further details.
In the case of a strictly positive mass at the origin, De Marco, Hillairet and Jacquier~\cite[Theorem 3.6]{DeMarco} proved that
$$
\Iv(x) = \sqrt{2|x|} + \mathfrak{c} + \varphi(x), 
\qquad\text{as }x\downarrow -\infty,
$$
where the constant~$\mathfrak{c}$ is related to the mass at zero and the function~$\psi$ tends to zero as~$x$ tends to infinity, 
which, while capturing the slope~$2$, is markedly different from our new formula here.
Before being able to prove the theorem, we need two lemmas providing bounds on prices of Put options 
and on the implied volatility.

\begin{lemma}\label{lem:PutBound}
Let $q \geq 0$ be such that $\EE\left[\left|\log \ts_T\right|^q\right]$ is finite.
Then for all $x <  (q-1)\ind_{\{q<1\}}$, 
$$
\Pbs(x, \Iv(x)) \leq \E^x|x|^{-q} \EE\left[\left|\log \ts_T\right|^q\right].
$$
\end{lemma}

\begin{proof}
 The case $q=0$ is a consequence of no-arbitrage bounds for the Put option.  
Now consider $q>0$. 
For ease of exposition only, we work in the moneyness unit $k = \E^x$.  
The map $k\mapsto|\log(k)|^q$ is strictly convex on $\mathcal{K}:=(0, \E^{q-1}\ind_{\{q<1\}} + \ind_{\{q \geq 1\}})$.
Let now $v_q(k)$ denote the solution the equation to 
\begin{equation}\label{eq:v_def}
k = v_q(k)  \left(1 - \frac{1}{q} \log v_q(k)\right), 
\end{equation} 
for $k\in\mathcal{K}$, such that $\lim_{k \downarrow 0}v_q(k)= 0$. 
This equation van be solved explicitly as
$$
v_q(k) = \exp\left\{\Ww_{-1}\left(-q k \E^{-q}\right) + q\right\}.
$$ 
Recall~\cite{Knuth} that the Lambert~$\Ww$ function is such that for $z\in\RR$,  $\Ww(z)$ solves $\Ww(z)\E^{\Ww(z)}=z$, 
which is multi-valued for $-\E^{-1} < z \leq 0$. 
The~$\Ww_{-1}$ branch is the one that satisfies $\lim_{z \uparrow 0}\Ww_{-1}(z) =  -\infty$.  
Then, for $u \in [0,k]$, the curves $k\mapsto\frac{1}{q} v_q(k)|\log v_q(k)|^{1-q}|\log u|^q$ and $k\mapsto(k - u)^+$ 
are equal and have the same gradient at $u = v_q(k)$.  
Since $k\in\mathcal{K}$, strict convexity and positivity of $|\log(u)|^q$ over~$(0,k)$ 
imply on taking expectations that
\begin{equation}\label{eq:lamb_bound}
\Pbs(\log(k), \Iv(\log(k))) \leq \frac{1}{q}
v_q(k)|\log v_q(k)|^{1-q} \EE\left[\left|\log \ts_T\right|^q\right].
\end{equation} 
By construction $0<v_q(k)<k\leq 1$, hence $|\log v_q(k)|^{-q} < |\log(k)|^{-q}$.  
In particular using $\log v_q(k) < 0$ and~\eqref{eq:v_def}, it holds $\frac{1}{q}v_q(k) |\log
v_q(k)| < k$.  Combining these, a larger bound (than in
\eqref{eq:lamb_bound}) for Put prices is
given by
\begin{equation}\label{eq:put_bound}
\Pbs(\log(k), \Iv(\log(k))) \leq k|\log(k)|^{-q} \EE\left[\left|\log \ts_T\right|^q\right].
\end{equation} 
\end{proof}

\begin{remark}
In the limit $x \downarrow-\infty$ (or $k\downarrow 0$), 
we are indifferent between the bounds in
\eqref{eq:put_bound} and~\eqref{eq:lamb_bound},
because
\begin{equation}
\lim_{k \downarrow 0} \frac{\frac{1}{q}v_q(k)|\log v_q(k)|^{1-q}}{k|\log
  k|^{-q}} = 1. \label{eq:same_lims} \end{equation}
To see this, first recall that $\lim_{k \downarrow 0} v_q(k) = 0$, then from
\eqref{eq:v_def},
$$ \lim_{k \downarrow 0} \frac{k}{\frac{1}{q}v_q(k) |\log v_q(k)|} = 1.$$
Further, taking logarithm of both sides of~\eqref{eq:v_def} it follows
that 
$\lim_{k \downarrow 0}\frac{\log(k)}{\log v_q(k)} = 1$, which implies~\eqref{eq:same_lims}.
\end{remark}

\begin{lemma}\label{lem:IVbound}
Let $q\geq 0$ such that $\EE\left[|\log \ts_T|^q\right]$ is finite.
Then for any $p \in [0,q)$, there exists $x_p < 0$ such that 
$$
\Iv(x) < \sqrt{-2x + 2p \log(|x|)} - \sqrt{2p \log(|x|)},
\qquad \text{for all }x < x_p.
$$
\end{lemma}

\begin{proof}
The case $q=0$ is clear from Lemma~\ref{applem:leebound}.
Let $q>0$. 
Note that when the implied volatility is of the form
$$\Iv(x) = \sqrt{2}(f(x) - g(x)),
\qquad\text{for }x\in\RR,$$ 
where $f, g:\RR\to\RR$ satisfy $f(x)^2  - g(x)^2 = -x$, the corresponding price of the Put option~\eqref{eq:Put}
is given by
$$
\Pbs(x,\Iv(x)) = \E^x \Phi\left(- \sqrt{2} g(x)\right) - \Phi\left(- \sqrt{2} f(x) \right).
$$
In our case, the two functions are given by
$f(x) = \sqrt{p \log(|x|)-x}$
and 
$g(x) = \sqrt{p \log(|x|)}$.
With~$\phi$ denoting the Gaussian density, the asymptotic relationship
\begin{equation}\label{eq:gauss_identity}
\lim_{z \uparrow\infty} \frac{z\Phi(-z)}{\phi(z)} = 1,
\end{equation}
holds trivially by L'H\^opital's rule and therefore
$$
\lim_{x \downarrow-\infty} \frac{\Phi \left(- \sqrt{2} f(x)\right)}{\E^x \Phi\left(-\sqrt{2}g(x)\right)} = 0,
$$
which implies
$$\lim_{x \downarrow-\infty} \frac{\Pbs\left(x,\sqrt{2}(f(x) - g(x)) \right)}{\E^x \Phi\left(-\sqrt{2} g(x)\right)} =1.$$
We can then deduce 
\begin{align}
\lim_{x \downarrow-\infty}\frac{\E^x |x|^{-q}}{\Pbs\left(x,\sqrt{2}(f(x) - g(x)) \right)}
 =  \lim_{x \downarrow-\infty}\frac{\E^x |x|^{-q}}{\E^x \Phi\left(-\sqrt{2}g(x)\right)}
& =  \lim_{x \downarrow-\infty}\frac{|x|^{-q}}{\Phi\left(-\sqrt{2}g(x)\right)}\\
& =  \lim_{x \downarrow-\infty}\frac{-\sqrt{2}g(x) |x|^{-q}}{\phi\left(-\sqrt{2}g(x)\right)}\\
& =  \lim_{x \downarrow-\infty}\frac{2\sqrt{\pi}g(x) |x|^{-q}}{\E^{-g(x)^2}}\\
& =  \lim_{x \downarrow-\infty}2\sqrt{\pi}g(x) |x|^{p-q}
 =  \left\{ \begin{array}{ll}
 0, & \textrm{if $p<q$},\\
 \infty, & \textrm{if $p \geq q$},
  \end{array} \right.  
\label{eq:IV_Bound_relation}
\end{align}
and the lemma follows from Lemma~\ref{lem:PutBound} and the monotonicity of $\Pbs(\cdot, \cdot)$ in its second argument.
\end{proof}

Before stating the proof of Theorem~\ref{thm:LMFormula},
recall the following lemma, which will be used repeatedly:
\begin{lemma}\label{lem:Replic}
For any convex function $f:\RR_+ \to \RR$, the identity
$$
f(x) = f(x_0) + f'(x_0)(x - x_0) + \int_0^{x_0}(y - x)^+ \mu(\D y) + \int_{x_0}^{\infty}(x - y)^+ \mu(\D y),
$$
holds for Lebesgue almost all $x,x_0 \in \RR_+$, where $\mu = f''$ in the sense of distributions. 
\end{lemma}

\begin{proof}[Proof of Theorem~\ref{thm:LMFormula}]
Let $\zeta:=\liminf_{x \downarrow-\infty}\frac{\df(x, \Iv(x))}{\sqrt{2\log|x|}}$.
Suppose by contradiction that $\zeta < \sqrt{ \qf}$ 
and let~$q$ such that $\sqrt{q} \in (\zeta, \sqrt{\qf})$. 
Then there exists a sequence $(x_n)_{n \in \mathbb{N}}$ with $x_n \downarrow-\infty$ such that for all $n$,
$\df(x_n, \Iv(x_n)) < \sqrt{2q \log |x_n|}$.
Inverting this yields
$$
\Iv(x_n) > \sqrt{-2x_n + 2q\log|x_n|} - \sqrt{2 q \log|x_n|},
$$
which contradicts Lemma~\ref{lem:IVbound} since $q<\qf$.
Assume now that $\zeta > \sqrt{\qf}$ and let~$q$ such that $\sqrt{q} \in (\sqrt{\qf}, \zeta)$.  
We show that this implies $\EE[|\log \ts_T|^p]$ is finite for all $p \in (\qf,q)$.
Indeed, in this case, 
$$
\Iv(x) < \sqrt{-2x + 2q\log(|x|)} - \sqrt{2q \log(|x|)}
$$
for all $x$ along a sequence. 
From~\eqref{eq:IV_Bound_relation} it follows that there exists $x^*<0$ such that for $x<x^*$,
\begin{equation}
\Pbs(x,\Iv(x)) < \E^x|x|^{-q}. \label{eq:conineq} \end{equation}
Now, reverting to moneyness units $k = \E^x$, one sees that for $p\in(\qf, q)$ and 
$z_p = \E^{p-1}\ind_{\{p<1\}} + \ind_{\{p \geq 1\}}$,
\begin{align*}
\EE\left[\left|\log \ts_T\right|^p\right]
 & = \EE \left[\left|\log \ts_T\right|^p\ind_{\{\ts_T < z_p\}}\right] +  |p-1|^p
 + \EE \left[\left|\log \ts_T\right|^p\ind_{\{\ts_T \geq z_p\}}\right] -  |p-1|^p \\
 & = \int_0^{z_p} \frac{\Pbs(\log(k), \Iv(\log(k)))}{k^2}p\left\{ (p-1)|\log(k)|^{p-2} +|\log(k)|^{p-1} \right\} \D k
 + \EE \left[|\log \ts_T|^p\ind_{\{\ts_T \geq  z_p\}}\right] -  |p-1|^p \\
 & \leq \int_0^{z_p}\frac{k}{|\log(k)|^{p}}\frac{p}{k^2}\left\{ (p-1)|\log(k)|^{p-2} +|\log(k)|^{p-1} \right\} \D k
 + \EE \left[\left|\log \ts_T\right|^p\ind_{\{\ts_T \geq z_p\}}\right] -  |p-1|^p < \infty.
\end{align*}
Since $x\mapsto|\log x|^p$ is strictly convex on $(-\infty, z_p)$, 
the second line above follows from Lemma~\ref{lem:Replic} applied to the convex function
$x\mapsto|\log x|^p\ind_{\{x < z_p\}} + |p-1|^p$ and taking expectation with
$x = \ts_T$.  
The third line follows from the strictly positive second derivative of $x\mapsto|\log x|^p$ 
in the interval considered and~\eqref{eq:conineq}.  
The final equality holds since $q>\qf$ and the second expectation on the first line is finite. 
\end{proof}

\subsection{Refinement of the Fukasawa-Gatheral formula}
In his volatility Bible~\cite{Gatheral}, Gatheral derived an elegant formula expressing the log contract
directly in terms of the implied volatility. 
This has obvious appeal as traders can plug in their favourite implied volatility smile (parametric or not) 
and obtain the fair value of a variance swap. 
Earlier versions of this formula, albeit with more sketchy proofs, were proposed by Matytsin~\cite{Matytsin} and Chriss and Morokoff~\cite{Chriss}.
A fully thorough derivation though has only recently been provided by Fukasawa~\cite{Fukasawa}
(see also~\cite{Lucic} for interesting connectiong with absence of arbitrage)
who not only proved the key ingredient, the decreasing property of the map $k\mapsto\df(k,\dot)$, 
but extended the formula to more general payoff contract.
In all these proofs, the main assumption is the existence of moments $\EE[\ts_T^{1+\eps}]$ for some $\eps>0$.
We show hereafter that this additional condition is in fact not required.
Following~\cite{Fukasawa}, let
\begin{equation}\label{eq:fg}
\ff(x) := -\df(x, \Iv(x)) = \frac{x}{\Iv(x)} + \frac{\Iv(x)}{2},
\end{equation}
and note that, as proved by Fukasawa~\cite{Fukasawa}, the inverse function~$\gf$ is well defined.
This yields the following:
\begin{theorem}\label{thm:GathFuka}
If $\EE\left[|\log(\ts_T)|\right]$ is finite (namely $\qf\geq 1$), then
$$
- 2 \EE[\log(\ts_T)] = \int_{\RR} \Iv(\gf(z))^2\phi(z) \D z.
$$ 
\end{theorem} 

\begin{proof}
Note first that, by~\cite[Theorem 2.8]{Fukasawa} the map $x \mapsto \df(x, \cdot)$ is decreasing.
By~\eqref{eq:Put} and the Put-Call parity, a Call option with log-moneyness $x = \log(K/F_T)$ is worth   $$
\Cbs(x,\sigma) = \Phi[\dfs+ \sigma] - \E^x \Phi[\dfs].
$$
By Lemma~\ref{lem:Replic}, with $c(\E^x) := \Cbs(x,\Iv(x))$ and $p(\E^x) := \Pbs(x,\Iv(x))$, we can write
\begin{align*}
\mathfrak{L} := \EE[-\log(\ts_T)]
&= \int_{-\infty}^0 p(\E^x)\E^{-x} \D x + \int_0^{\infty} c(\E^x)\E^{-x} \D x \\
&= \left[-p\left(\E^x\right)\E^{-x} \right]^0_{-\infty} + \left[-c\left(\E^x\right)\E^{-x}\right]^{\infty}_0
 + \int_{-\infty}^0p'(\E^x)\D x + \int_0^{\infty}c'(\E^x)\D x \\
&= \int_{-\infty}^0p'(\E^x)\D x +  \int_0^{\infty}c'(\E^x)\D x.
\end{align*}
The boundary terms vanish because $c(1) = p(1)$ by Put-Call parity, 
because~$c(\cdot)$ tends to zero for large strikes and by Lemma~\ref{lem:PutBound} 
since $p(x) \leq \E^x|x|^{-1}\EE[|\log\ts_T|]$ for $x<0$ implies
$\lim_{x \downarrow-\infty} p(\E^x)\E^{-x} = 0$.
Now, 
$$p'(\E^x)\E^x = \frac{\D}{\D x}\Pbs(x,\Iv(x))
\qquad \text{and} \qquad
c'(\E^x)\E^x = \frac{\D}{\D x}\Cbs(x,\Iv(x)).$$
Hence, with $\delta(x):=\df(x, \Iv(x))$,
\begin{equation}
\label{eq:deriv}
\begin{split}
p'(\E^x) =& \Phi[-\delta(x)] - \phi(\delta(x))\delta'(x) + \E^{-x}\phi(-\delta(x) - \Iv(x))[\delta'(x) + \Iv'(x)]\\
c'(\E^x)= &  \E^{-x}\phi(\delta(x) +\Iv(x))[\delta'(x) + \Iv'(x)] -\Phi[\delta(x)] - \phi(\delta(x))\delta'(x).
\end{split} \end{equation}
Since the Gaussian density~$\phi$ satisfies $\phi(a+b) = \phi(a-b)\E^{-2ab}$ for any $a,b\in\RR$, then
$$
\E^{-x}\phi(\delta(x) +  \Iv(x)) = \E^{-x}\phi\left( -\frac{x}{\Iv(x)} + \frac{\Iv(x)}{2} \right) = \phi(\delta(x)),
$$
and hence the system~\eqref{eq:deriv} simplifies, by symmetry of~$\phi$, to
$$
p'(\E^x) = \Phi[-\delta(x)]+ \phi(\delta(x))\Iv'(x)
\qquad\text{and}\qquad
c'(\E^x) = \phi(\delta(x))\Iv'(x) -\Phi[\delta(x)].
$$
Therefore
\begin{align*}
  \mathfrak{L} &= \int_{-\infty}^0 \Phi[-\delta(x)] \D x - \int_0^{\infty}
  \Phi[\delta(x)] \D x + \int_{\RR}\phi(\delta(x))\Iv'(x) \D x\\
&= \left[ x \Phi[-\delta(x)] \right]^0_{-\infty} - \left[x \Phi[\delta(x)]
    \right]^{\infty}_0 
+ \int_{\RR}x \phi(\delta(x))\delta'(x) \D  x + \int_{\RR}\phi(\delta(x))\Iv'(x) \D x.
\end{align*}
For the boundary terms, observe first from the log-moment formula, Theorem~\ref{thm:LMFormula}, that $\qf \geq 1$ implies $\delta(x) \geq \sqrt{2\log|x|}$ eventually for $x < 0$, and so
$\exp\left\{-\frac{1}{2}\delta^2(x)\right\} \leq |x|^{-1}$.
Combining with the identity
\eqref{eq:gauss_identity} one sees $\lim_{x \downarrow-\infty} x\Phi[-\delta(x)]=0$.  
Now, Lemma~\cite[Lemma 3.1]{Lee} 
(the right-tail analogue of Lemma~\ref{applem:leebound}), 
implies the trivial bound $\Iv(x) \leq \sqrt{2x}$ for $x>0$ sufficiently large
and therefore 
$$
\delta(x) = -\left(\frac{x}{\Iv(x)} + \frac{\Iv(x)}{2} \right) \leq -\frac{x}{\Iv(x)} \leq -\frac{\sqrt{x}}{2},
$$
which diverges to $-\infty$ as~$x$ tends to infinity. 
Therefore, for~$x$ large enough,
$$
0\leq \frac{x\phi(-\delta(x))}{-\delta(x)}
 = \frac{1}{\sqrt{2\pi}}\frac{x \exp\left\{-\frac{1}{2}\delta^2(x)\right\}}{-\delta(x)}
 \leq \frac{1}{\sqrt{2\pi}}\frac{x  \exp\left\{-\frac{1}{4}x\right\}}{\frac{x}{\Iv(x)}}
 = \frac{\Iv(x) \exp\left\{-\frac{1}{4}x\right\}}{\sqrt{2\pi}}
\leq \frac{\sqrt{x}\exp\left\{-\frac{1}{4}x\right\}}{\sqrt{\pi}},
$$
which tends to zero as~$x$ tends to infinity.
The limit~\eqref{eq:gauss_identity} thus implies
$\lim_{x \uparrow\infty} x\Phi[\delta(x)]=0$ and therefore
\begin{align*}
  \mathfrak{L}
 & = \int_{\RR}x \phi(\delta(x))\delta'(x) \D x + \int_{\RR}\phi(\delta(x))\Iv'(x) \D x\\
 & = \int_{\RR} x \phi(\delta(x))\delta'(x) \D  x 
 + \left[\Iv(x)\phi(\delta(x)) \right]_{\RR}  + \int_{\RR}\phi \left(\delta(x)\right) \delta(x) \delta'(x)\Iv(x) \D x\\
 &  = \int_{\RR} \phi(\delta(x))\delta'(x)[x + \Iv(x)\delta(x)]\D x
= - \int_{\RR} \phi(\delta(x))\delta'(x)\frac{\Iv^2(x)}{2}  \D x,
\end{align*}
where the boundary terms cancel as above and by Lemma~\ref{applem:leebound}, 
and applying~\eqref{eq:d} for $\delta(x)$.
Substituting $z = \delta(x)$, using the symmetry of~$\phi$, 
the proposition follows from the limits $\lim_{x \to \pm \infty}\delta(x) = \mp \infty$.
\end{proof}

\subsection{Pricing formulae for European options}\label{sec:FukExtension}
In~\cite{Fukasawa}, Fukasawa not only proved a version of Theorem~\ref{thm:GathFuka} (with more restrictive assumptions), 
but also extended it to options with payoffs of the form
$\Psi(\log(S_T))$ for any twice differentiable function~$\Psi$ with derivative of at most polynomial growth.
More precisely, he derived~\cite[Theorem 4.4]{Fukasawa} an integral form for $\EE\left[\Psi(\log(S_T))\right]$
assuming either that $\EE[S_T^{1+p}]$ exists for some $p>0$ or that $\EE[S_T^{-q}]$ exists for some $q>0$.
The former case is not affected by our setup and we instead provide a refinement of the latter case when no such~$q$ exists
but instead log-moments are available.
This in fact extends the scope of Theorem~\ref{thm:GathFuka} above.
Recall that the function~$\ff$ is defined in~\eqref{eq:fg}
and let $\Pp_{q}$ denote the set of functions with at most polynomial growth of order~$q$ at $-\infty$.

\begin{theorem}
Assume that
$\qf := \sup\left\{q\geq 0: \EE\left[\left|\log \ts_T\right|^q\right] < \infty\right\}$ belongs to $[1, \infty)$.
\begin{itemize}
\item For any twice differentiable function~$\Psi \in \Pp_{q}$ with $q \in [0, \qf]$,
$$
\EE\left[\Psi(\log(S_T))\right]
 = \int_{\RR}\left\{\Psi(\gf(z)) - \Psi'(\gf(z))\left[\gf(z)+\frac{\Iv\left(\gf(z)\right)^2}{2}\right]\right\}\phi(z)\D z
 + \int_{\RR}\Psi''(x)\Iv(x) \phi(\ff(x))\D x.
$$
\item  For any absolutely continuous function~$\Psi \in \Pp_{q}$ with $q \in [0, \qf]$,
$$
\EE\left[\Psi(\log(S_T))\right]
 = \int_{\RR}\left\{\Psi(\gf(z)) - \Psi'(\gf(z)) + \Psi'(\hf(z))\E^{-\hf(z)}\right\}\phi(z)\D z,
$$
where~$\hf$ is the inverse function of the map $x\mapsto \ff(x) - \Iv(x)$.
\end{itemize}
\end{theorem}
\begin{remark}
With $\Psi(x) \equiv x$, then $\Psi\in\Pp_{1}$ and $\Psi'\in\Pp_{0}$, proving Theorem~\ref{thm:GathFuka}.
\end{remark}
\begin{proof}
The proof of this theorem follows that of~\cite[Theorem~4.4]{Fukasawa}, or indeed that of Theorem~\ref{thm:GathFuka} above.
The steps are analogous, but one has to pay special attention to the boundary terms arising from the different integrations by parts involved.
In our setting, the two terms that need special care are 
\begin{equation}\label{eq:IBPBoundaryTerms}
\lim_{x\downarrow-\infty}\Psi'(x)\Iv(x)\phi(\ff(x))
\qquad\text{and}\qquad
\lim_{x\downarrow-\infty}\Psi(x)|\Iv'(x)|\phi(\ff(x)),
\end{equation}
which we need to send to zero for a suitable class of functions~$\Psi$.

By Theorem~\ref{thm:LMFormula}, 
$\sqrt{\qf}$ is the largest value such that for any $\eps>0$, there exists $x_\eps$ for which
\begin{equation}\label{eq:IneqSigma}
\frac{\df(x, \Iv(x))}{\sqrt{2 \log|x|}} > \sqrt{\qf} - \eps =:\sqrt{\qfe},
\end{equation}
for all $x\leq x_\eps$.
Now, the equation (in~$\sigma$) $\frac{\df(x, \sigma)}{\sqrt{2 \log|x|}} = \sqrt{\qfe}$ admits two roots
$\sigma_{\pm} =  - \sqrt{2\qfe\log(|x|)} \pm \sqrt{2\qfe\log(|x|) - 2x}$,
so that, for $x<x_\eps$, the inequality~\eqref{eq:IneqSigma} holds if (similarly to Lemma~\ref{lem:IVbound} in fact)
\begin{equation}\label{eq:UpperBound_Iv}
\Iv(x) < - \sqrt{2\qfe\log(|x|)} + \sqrt{2\qfe\log(|x|) - 2x}.
\end{equation}
Note that when $\qf=0$ and replacing the $\liminf$ by a genuine limit, this reads $\Iv(x) < \sqrt{2|x|}$ for~$x$ small enough, 
which was proved by Lee~\cite{Lee}.
This further implies directly that for $x<x_\eps$, 
\begin{equation}\label{eq:UpperBound_f}
\ff(x) < -\sqrt{2\qfe\log(|x|)}.
\end{equation}
Therefore for any function~$\Psi: (-\infty, x_\eps]\to\RR$,
$$
\Psi'(x)\Iv(x)\phi(\ff(x)) = 
\frac{\Psi'(x)\Iv(x)}{\sqrt{2\pi}}\exp\left\{-\frac{\ff(x)^2}{2}\right\} \leq 
\frac{\Psi'(x)\Iv(x)}{\sqrt{2\pi}}\E^{-\qfe\log(|x|)} = 
\frac{\Psi'(x)\Iv(x)}{\sqrt{2\pi}}|x|^{-\qfe}.
$$
From the bound~\eqref{eq:UpperBound_Iv} on~$\Iv(x)$, this expression tends to zero as~$x\downarrow -\infty$ if and only if
$\Psi'\in\Pp_{q'}$ with $q' \in [0, \qf - \frac{1}{2}]$.
Clearly when $\qf \in [0,\frac{1}{2}]$, this cannot tend to zero as~$\Iv(x)$ dominates~$\phi(\ff(x))$.
This refines the analysis of~\cite[Lemma 4.2]{Fukasawa} which assumed the existence of strictly negative moments for the stock price.
Now Fukasawa showed~\cite[Lemma 2.6]{Fukasawa} that, independently of any moment (or log-moment) assumptions, 
$\ff(x)\Iv'(x) <1$ for all $x\in\RR$;
combining this with the new upper bound~\eqref{eq:UpperBound_f}, 
we obtain a new version of~\cite[Theorem~3.6]{Fukasawa}, namely 
$$
\Iv'(x) > -\frac{1}{\sqrt{2\qf\log(|x|)}},
$$
for $x$ small enough, so that 
$|\Iv'(x)| < \left(2\qf\log(|x|)\right)^{-1/2}$ and therefore 
$$
\Psi(x)|\Iv'(x)|\phi(\ff(x)) = 
\frac{\Psi(x)|\Iv'(x)|}{\sqrt{2\pi}}\exp\left\{-\frac{\ff(x)^2}{2}\right\} \leq 
\frac{\Psi(x)|\Iv'(x)|}{\sqrt{2\pi}}\E^{-\qfe\log(|x|)} = 
\frac{\Psi(x)}{2\sqrt{\pi\qf}}\frac{|x|^{-\qfe}}{\sqrt{\log(|x|)}}
$$
converges to zero as $x\downarrow-\infty$ as soon as~$\Psi\in\Pp_q$ with $q \in [0,\qf]$.
This therefore implies that the two limits~\eqref{eq:IBPBoundaryTerms} are equal to zero
if and only if $\Psi\in\Pp_q$  for $q \in [0,\qf]$.
All the other statements in~\cite[Lemma 4.3]{Fukasawa} remain identical, 
and therefore the proof of Theorem~4.4 follows analogously, 
the boundary terms cancelling out under our new assumptions, 
thus proving the first bullet point in the theorem.
Indeed, the two conditions are that
$\Psi\in\Pp_q$  for $q \in [0,\qf]$ and $\Psi' \in \Pp_{q'}$
with $q' \in [0, \qf-\frac{1}{2}]$; 
the intersection of both is in fact the same as the former.
A close look at the proof of the second bullet point in~\cite[Theorem 4.4]{Fukasawa} shows that only the second limit
in~\eqref{eq:IBPBoundaryTerms} needs to tend to zero, which, as just discussed, is true as soon as $\Psi \in \Pp_{\qf}$, and
the theorem follows.
\end{proof}

\section{Examples}\label{sec:applications}
Corollary~\ref{corL:IVExpansion} gives us a recipe to estimate~$\qf$ (whenever it exists)
from market data by simple regression of the implied volatility against the log-moneyness.
This also facilitates informed initial guesses for model calibration, 
with a direct relationship between model parameters and the number of log-moments of the stock price admits.
We provide several examples of models where this is feasible.

\subsection{Exponential L\'evy models}
In exponential L\'evy models the stock-price process is modelled by
\begin{equation}\label{eq:explevy}
S_t =  S_0\exp(L_t), 
\end{equation}
where $(L_t)_{t \in \TT}$ is a real-valued L\'evy process~\cite[Chapter 3]{Sato},
namely a c\`adl\`ag stochastically continuous process with independent and identically distributed increments
starting from $L_0 = 0$. 
For any $t >0$, the characteristic function of the random variable~$L_t$ satisfies
$$
\log\EE\left[\E^{\I u L_t}\right] = \psi(u)t,
$$
for all $u \in \RR$, where the characteristic exponent~$\psi$ admits the L\'evy-Khintchine representation
\begin{equation}\label{eq:lk}
\psi(u) = -\frac{\xi u^2}{2} + \I \gamma u + \int_{\RR} \left(\E^{\I u x} - 1 - \I ux \ind_{\{|x|\leq 1\}} \right) \nu(\D x),
\end{equation}
with $\xi \geq 0$, $\gamma\in\RR$ and~$\nu$ a measure on~$\RR$ satisfying $\nu(\{0\})=0$ and $\int_{\RR}(1 \wedge x^2) \nu(\D x) < \infty$.
Sato~\cite[Theorem~25.3]{Sato} proved that for any submultiplicative, locally bounded function~$g$,
the expectation $\EE[g(S_T)]$ is finite if and only if $\int_{\RR}g(x)\nu(\D x)$ is finite.
In light of Theorem~\ref{thm:LMFormula}, we thus consider the function $g(x) \equiv \log(|x|)^q$ with $q\geq 0$.

\subsubsection{Finite moment log stable process}
The Finite Moment Log Stable (FMLS) model was introduced by Carr and Wu~\cite{CarrWu} 
to capture the observed negative skew observed on S\&P options.
There the driving L\'evy process~$L$ in~\eqref{eq:explevy} is $\alpha$-stable 
with tail index $\alpha \in (1,2)$ and skew parameter $\beta=-1$,
so that~\cite[Chapter 3]{Sato}, for any $T>0$,
\begin{itemize}
\item $\EE\left[|S_T|^{p}\right]$ is finite for all $p\geq 0$;
\item the support of~$L_T$ is the whole real line;
\item $\EE\left[|\log S_T|^q\right]$ is finite for all $q \in (0, \alpha)$ and is infinite if $q \geq \alpha$.
\end{itemize}
Theorem~\ref{thm:LMFormula} thus applies with $\qf = \alpha \in (0,2)$
and $\EE\left[|\log(S_T)|^2\right]$ is infinite.  
While the model may capture the fat left tail and thin right tail of the stock price, 
it is too extreme if a discrete variance swap is traded. 

\subsubsection{Finite moment log mixture model}
In~\eqref{eq:explevy} let $L:=X-Y$ for two independent processes~$X$ and~$Y$ with

\begin{itemize}
\item $\qf_X := \sup\{q\geq 0: \EE\left[|X_1|^q\right]<\infty\}>0$ and $\EE\left[\E^{\pf_X X_1}\right]$ is finite for some $\pf_X\geq 1$;
\item $\qf_Y :=
  \sup\{q\geq0: \EE\left[|Y_1|^q\right]<\infty\}\in (0,\qf_X)$ and $\EE\left[\E^{-\pf_Y Y_1}\right]$ for some $\pf_Y \in[1,\pf_X)$,
\end{itemize}
so that $X$ and $Y$ respectively influence the right and left tails in the distribution. 
Before identifying some candidates for the process $X$ and $Y$, we note:

\begin{lemma}
$\EE\left[\E^{\pf_Y L_1}\right]$ is finite and $\qf_L :=  \sup\{q\geq0: \EE\left[|L_1|^q\right]<\infty\} = \qf_Y$.
\end{lemma}
\begin{proof}
The first statement follows by independence of~$X$ and~$Y$, 
so that the moment generating function of~$L$ is simply the product of those of~$X$ and~$Y$.
Now, it is clear that $\EE|L_1|^q$ is finite for $q<\qf_Y$. For
$q>\qf_Y$, observe
$$
|Y_1|^q \leq \left|\left(|Y_1| - |X_1|\right)^{+} + |X_1| \right|^q < 2^q \left(\left\{\left(|Y_1| - |X_1|\right)^{+}\right\}^q + |X_1|^q \right)
$$
and $(|Y_1| - |X_1|)^+ \leq \left| |Y_1| - |X_1| \right| \leq |Y_1 - X_1|$, where this last inequality is due to the reverse triangular
inequality. This implies the assertion about~$\qf_L$. 
\end{proof}

Choices for~$X$ abund, as any process with finite moments and finite exponential moments of all orders will do, 
in particular the Brownian motion, the generalised Inverse Gaussian process,
the generalised Hyperbolic process~\cite{BNS}, the CGMY process~\cite{CGMY}.
For~$Y$, the choices are scarcer, but the inverse Gaussian process is a valid one, 
whereby $Y$ is a pure-jump L\'evy process with density at time~$1$ equal to
$$
f_{\mathrm{IG}}(y; \alpha, \beta) =
\frac{\beta^{\alpha}}{\Gamma(\alpha)}y^{-\alpha-1}\E^{-\beta/y},\quad
\textrm{for $y>0$},
$$
where $\alpha, \beta>0$ are the shape and scale parameters and $\Gamma(\cdot)$ is the Gamma function.
J{\o}rgensen~\cite{Jorgensen} showed that 
$$
\EE\left[Y^r\right] = 
\frac{\Gamma(\alpha - r)}{\Gamma(\alpha)}\beta^r, \quad \text{if }r<\alpha,
\text{ and infinite otherwise}.
$$ 
The reciprocal Gamma distribution is a special case of the Generalised Inverse Gaussian $(\mathrm{GIG})$ distribution 
and hence is infinitely divisible~\cite{BNS}.  
With this specification, the log-returns have exploding negative moments beyond order $\qf_L=\alpha$
(possibly larger than~$2$) and positive moments of  arbitrary order depending on~$X$.

\subsection{Stochastic volatility models}
The final example we are interested in belongs to the class of classical stochastic volatility models, 
where~$S$ satisfies the following dynamics under the risk-neutral probability measure:
\begin{equation*}
\begin{array}{rl}
\D S_t & = \sigma_t^{\delta} S_t \left(\rho\,\D W_t + \sqrt{1-\rho^2}\D W_t^{\perp}\right),\\
\D \sigma_t & = b(\sigma_t)\D t + \nu \sigma_t^\gamma \D W_t,
\end{array}
\end{equation*}
starting from $S_0, \sigma_0>0$, where $\rho \in [-1,1]$, $\delta, \gamma, \nu>0$ and $b(\cdot)$ is some drift.
Lions and Musiela~\cite{Lions} provided necessary (and often sufficient) conditions 
on the parameters and the drift ensuring that~$S$ is a true martingale and that moments of a certain order exist.
A particularly interesting case was recently highlighted by Carr and Willems~\cite{CarrWill} with the specifications
$\delta = \gamma = 1$
and 
$b(\sigma) = (R_0+R_1\sigma)(R_2-\sigma)$,
with $R_0, R_1\geq 0$ and $R_2>0$.
Using~\cite{Lions}, they showed that for any $\rho \in [-1,0]$,
Roger Lee's largest negative moment is actually equal to $\pf=0$.
We leave it to future endeavours to compute the precise value of~$\qf$.


\end{document}